\documentclass[11pt]{article}
\usepackage[T1]{fontenc}
\usepackage[utf8]{inputenc}
\usepackage{amsmath,amssymb,amsthm}
\usepackage{mathtools}
\usepackage{microtype}
\usepackage{aliascnt}
\usepackage[hidelinks]{hyperref}
\newtheorem{theorem}{Theorem}[section]
\newaliascnt{lemma}{theorem}
\newtheorem{lemma}[lemma]{Lemma}
\aliascntresetthe{lemma}
\newaliascnt{corollary}{theorem}
\newtheorem{corollary}[corollary]{Corollary}
\aliascntresetthe{corollary}

\usepackage[nameinlink,noabbrev]{cleveref}
\crefname{lemma}{lemma}{lemmas}
\Crefname{lemma}{Lemma}{Lemmas}
\crefname{corollary}{corollary}{corollaries}
\Crefname{corollary}{Corollary}{Corollaries}
\crefname{remark}{remark}{remarks}
\Crefname{remark}{Remark}{Remarks}

\newcommand{\paperauthor}{Hangyu Liu}
\newcommand{\paperaffiliation}{Independent Researcher}
\newcommand{\paperemail}{ylzz1997@outlook.com}
\newcommand{\Tset}{\mathcal T}
\newcommand{\Rtree}{\mathcal R}
\newcommand{\eps}{\varepsilon}
\newcommand{\ip}[2]{\left\langle #1,#2\right\rangle}

\title{A Paturi Theorem for Signed Subcube Representations}
\author{\paperauthor\\\paperaffiliation\\\texttt{\paperemail}}
\date{}

\begin{document}
\maketitle

\begin{abstract}
For a symmetric Boolean function $f_F(x)=F(|x|)$, let
\[
 D(F)=\max_{t:F(t-1)\ne F(t)}\min\{t,n-t+1\}
\]
be the depth of its deepest transition.  A generalized monomial is a subcube
indicator.  We prove the quantitative upper bounds
\[
\begin{aligned}
 \widetilde{\operatorname{gwt}}_\eps(f_F)
 &\le 2^{O(D(F)\log(2/\eps))},\\
 \widetilde{\operatorname{gspar}}_\eps(f_F)
 &\le 2^{O(D(F)\log(2/\eps))}\log(n+1).
\end{aligned}
\]
For every fixed $0<\eps<1/2$, the first bound is optimal:
\[
 \widetilde{\operatorname{gwt}}_\eps(f_F)
 =2^{\Theta_\eps(D(F))}
\]
for every nonconstant symmetric function, uniformly in $n$ and without any
restriction on its transition pattern.  For $D(F)\ge2$, we also prove
\[
 \log\!\left(\widetilde{\operatorname{gspar}}_\eps(f_F)+1\right)
 =\Theta_\eps\!\left(D(F)+\log\log(n+2)\right),
\]
whereas $D(F)=1$ implies exact sparsity at most three.

The lower bound uses an exponential restriction-profile transfer and a
leaf-dependent completion that centers a deepest transition.  Its dual form
lifts approximate-degree witnesses from the leaves and yields a signed
measure with constant target correlation but
$2^{-\Omega(D(F))}$ correlation with every subcube indicator.  For the upper
bound, every exponential profile $r^w$ has generalized weight at most one.
A one-shot Jackson approximation handles an arbitrary edge pattern at
constant error; a ternary output amplifier and an empirical sparsification
lemma give the stated quantitative bounds.  As a consequence of Paturi's
theorem and quantum approximate counting,
\[
 \log\!\left(\widetilde{\operatorname{gwt}}_\eps(f_F)+1\right)
 =\Theta_\eps\!\left(\frac{Q_{1/3}(f_F)^2}{n}\right).
\]
\end{abstract}

\section{Introduction}\label{sec:intro}

Real polynomial representations of Boolean functions are measured not only
by degree, but also by the number and total coefficient mass of their
monomials.
Both quantities depend on the chosen dictionary.
This paper studies the dictionary of all generalized monomials
\[
 m_{P,N}(x)=\prod_{i\in P}x_i\prod_{j\in N}(1-x_j),
 \qquad P\cap N=\varnothing,
\]
which are exactly the indicators of subcubes of the Boolean cube.
We ask how economically a Boolean function can be approximated pointwise by
a signed linear combination of these indicators.

Symmetric functions already exhibit two distinct sources of complexity.
Write $f_F(x)=F(|x|)$.
The first is the depth $D(F)$ of the deepest transition of the univariate
profile $F$.
An edge predicate such as \textsc{Or} has $D(F)=1$, while majority, parity,
and a central exact layer have $D(F)=\Theta(n)$.
The second source is the ambient dimension itself: even when $D(F)$ is a
fixed integer at least two, a sparse representation must still distinguish
many coordinates.

Our first result is a complete classification of approximate generalized
weight:
\begin{equation}
 \widetilde{\operatorname{gwt}}_\eps(f_F)
 =2^{\Theta_\eps(D(F))}
 \qquad(0<\eps<1/2\text{ fixed}).
\label{eq:intro-weight}
\end{equation}
The constants are uniform in $n$ and in the number and arrangement of the
transitions of $F$.
Quantitatively, the upper bound is
$2^{O(D(F)\log(2/\eps))}$.
Thus approximation removes all ambient-dimension dependence from
generalized weight as soon as the error is positive.

For generalized sparsity, an ambient-dimension term is unavoidable.
We prove
\begin{equation}
 \log\!\left(\widetilde{\operatorname{gspar}}_\eps(f_F)+1\right)
 =\Theta_\eps\!\left(D(F)+\log\log(n+2)\right)
\label{eq:intro-sparsity}
\end{equation}
whenever $D(F)\ge2$.
The upper bound before taking logarithms is
\[
 2^{O(D(F)\log(2/\eps))}\log(n+1).
\]
The dimension term is witnessed by a coordinate-signature argument:
if an approximant uses $s$ generalized monomials, every coordinate receives
a signature in $\{+,-,0\}^s$, and a signature class larger than $D(F)$
cannot distinguish the two layers adjacent to a deepest transition.
Hence $n\le D(F)3^s$.
This also reveals a sharp change at $D(F)=1$, where every symmetric function
has an exact three-term representation.

\paragraph{Lower-bound method.}
Chattopadhyay, Dahiya, and Lovett introduced restriction trees for sparsity
and used them to relate exact and approximate generalized sparsity and
weight~\cite{ChattopadhyayDahiyaLovett2026}.
For a generalized monomial $M$ and a random leaf $\rho$ of a two-sided
restriction tree, we use the exponential potential
\[
 Z_\rho(M)=
 2^{\deg(M|_\rho)}\mathbf 1[M|_\rho\not\equiv0].
\]
The three branches of a queried literal contribute $0,1,2$, so
$\mathbb E_\rho Z_\rho(M)\le1$.
We query $D(F)$ coordinates and use the leaf-dependent completion to center
one deepest transition in the residual symmetric function.
Paturi's theorem~\cite{Paturi1992} then supplies approximate degree linear
in the number of free variables.
Taking an exponential moment yields the lower side of
\eqref{eq:intro-weight} and \eqref{eq:intro-sparsity}.

We also give a dual form of this transfer.
An approximate-degree dual witness at each leaf can be lifted and averaged
with weight $2^{k(\rho)}$.
The one-monomial estimate makes the resulting global witness bounded on
every subcube indicator.
Using explicit symmetric dual polynomials such as those of Bun and
Thaler~\cite{BunThaler2015}, this gives an explicit signed measure with
constant target correlation and exponentially small subcube discrepancy.
Equivalently, it yields an agnostic labeled distribution on which the target
symmetric predicate has constant edge but every conjunction has
exponentially small edge.

\paragraph{Upper-bound method.}
After symmetrization, generalized weight is an atomic norm generated by
hypergeometric subcube profiles.
The unit ball contains every exponential sequence $r^w$.
To approximate an arbitrary edge pattern, set $z=e^{-w/D}$ and interpolate
the desired values on the nodes $z_w=e^{-w/D}$.
The nodes are $\Omega(1/D)$ apart, so the interpolant is $O(D)$-Lipschitz
regardless of how often the profile oscillates.
Jackson approximation and a Chebyshev coefficient estimate give a
constant-error signed exponential sum of mass $2^{O(D)}$ for the entire
edge profile at once.

Applying Jackson approximation directly at error $\eps$ would give the
weaker dependence $2^{O(D/\eps)}$.
Instead we amplify the output.
The edge corrections take values in $\{-1,0,1\}$, and a majority polynomial
separates constant neighborhoods of these three values with degree
$O(\log(2/\eps))$.
Products of exponential profiles remain exponential profiles, so the
amplified mass is $2^{O(D\log(2/\eps))}$.
Finally, a probabilistic construction replaces every profile $r^{|x|}$ by a
short average of actual subcube indicators.
Only inputs on the first $O(\log(1/\delta)/(-\log r))$ layers require a union
bound; monotonicity controls every later layer.
This is the source of the single factor $\log(n+1)$ in generalized
sparsity.

\paragraph{A query-complexity form.}
Paturi's characterization gives
$\widetilde{\deg}_{1/3}(f_F)=\Theta(\sqrt{nD(F)})$.
The polynomial method and quantum approximate counting give the same formula
for bounded-error quantum query complexity.
Consequently,
\[
 \log\!\left(\widetilde{\operatorname{gwt}}_\eps(f_F)+1\right)
 =\Theta_\eps\!\left(Q_{1/3}(f_F)^2/n\right).
\]
Thus, within the symmetric class, signed-subcube weight is the exponential
of squared quantum query complexity per input coordinate.

\paragraph{Prior work.}
Nisan and Szegedy initiated the systematic study of exact and approximate
degree~\cite{NisanSzegedy1994}, and Paturi characterized approximate degree
for symmetric Boolean functions~\cite{Paturi1992}.
Sherstov developed approximation and inclusion--exclusion constructions for
arbitrary symmetric predicates~\cite{Sherstov2009}.
The polynomial method in quantum query complexity is due to Beals,
Buhrman, Cleve, Mosca, and de Wolf~\cite{BealsEtAl2001}; we use the
approximate-counting algorithm of Brassard, H{\o}yer, Mosca, and
Tapp~\cite{BrassardHoyerMoscaTapp2002}.

Most directly, the present work builds on the two-sided restriction-tree
framework of Chattopadhyay, Dahiya, and
Lovett~\cite{ChattopadhyayDahiyaLovett2026}.
Our contributions are the exponential-profile transfer and its dual form,
the leaf-dependent centering of a deepest symmetric transition, and the
matching upper-bound machinery specific to the symmetric subcube dictionary.
In particular, the one-shot edge approximation removes any dependence on the
number of transitions.

\section{Definitions and main results}\label{sec:results}

All logarithms are to base two unless a natural logarithm is written as
$\ln$.
Constants in $O_\eps,\Omega_\eps,\Theta_\eps$ may depend on the fixed
approximation error and on no other parameter.

For $g:\{0,1\}^n\to\mathbb R$, define
\[
 \operatorname{gspar}(g)
 =\min\left\{s:
 g=\sum_{j=1}^s a_jm_{P_j,N_j},\ a_j\ne0\right\},
\]
where the empty sum is allowed, and
\[
 \operatorname{gwt}(g)
 =\min\left\{\sum_j|a_j|:
 g=\sum_j a_jm_{P_j,N_j}\right\}.
\]
For $0\le\eps<1/2$, put
\[
 \widetilde{\operatorname{gspar}}_\eps(g)
 =\min_{\|p-g\|_\infty\le\eps}\operatorname{gspar}(p),
 \qquad
 \widetilde{\operatorname{gwt}}_\eps(g)
 =\min_{\|p-g\|_\infty\le\eps}\operatorname{gwt}(p).
\]

For a Boolean function $h$ and $\eta\ge0$, its $\eta$-approximate degree is
\[
 \widetilde{\deg}_\eta(h)
 =\min\{\deg q:\|q-h\|_\infty\le\eta\},
\]
where polynomials may be taken multilinear on the Boolean cube.

For $F:\{0,1,\ldots,n\}\to\{0,1\}$, write $f_F(x)=F(|x|)$ and define
\[
 \Tset(F)=\{t\in[n]:F(t-1)\ne F(t)\},
\]
\[
 D(F)=\max_{t\in\Tset(F)}\min\{t,n-t+1\},
\]
with $D(F)=0$ when $F$ is constant.

\begin{theorem}[Complete symmetric weight classification]
\label{thm:weight-classification}
There is an absolute constant $C>0$ such that, for every
$0<\eps<1/2$, every symmetric Boolean function satisfies
\begin{equation}
 \widetilde{\operatorname{gwt}}_\eps(f_F)
 \le 2^{CD(F)\log(2/\eps)}.
\label{eq:weight-upper}
\end{equation}
For every fixed $0<\eps<1/2$, there are constants
$c_\eps>0$ and $C_\eps\ge0$ such that every nonconstant symmetric Boolean
function satisfies
\begin{equation}
 \widetilde{\operatorname{gwt}}_\eps(f_F)
 \ge 2^{c_\eps D(F)-C_\eps}.
\label{eq:weight-lower}
\end{equation}
Consequently,
\[
 \widetilde{\operatorname{gwt}}_\eps(f_F)
 =2^{\Theta_\eps(D(F))}
\]
for every nonconstant symmetric $F$, uniformly in $n$ and in its transition
pattern.
\end{theorem}

\begin{theorem}[Sparse symmetric upper bound]
\label{thm:sparsity-upper}
There is an absolute constant $C>0$ such that every symmetric Boolean
function satisfies
\begin{equation}
 \widetilde{\operatorname{gspar}}_\eps(f_F)
 \le 2^{CD(F)\log(2/\eps)}\log(n+1).
\label{eq:sparsity-upper}
\end{equation}
The approximant establishing \eqref{eq:sparsity-upper} simultaneously has
generalized weight at most $2^{CD(F)\log(2/\eps)}$.
\end{theorem}

\begin{theorem}[Coordinate-signature lower bound]
\label{thm:signature-lower}
Let $0\le\eps<1/2$, and let $F$ be symmetric with $D=D(F)\ge2$.
Then
\begin{equation}
 \widetilde{\operatorname{gspar}}_\eps(f_F)
 \ge\log_3(n/D).
\label{eq:signature-lower}
\end{equation}
\end{theorem}

\begin{corollary}[Logarithmic-scale sparsity classification]
\label{cor:log-sparsity}
For every fixed $0<\eps<1/2$ and every nonconstant symmetric $F$ with
$D(F)\ge2$,
\begin{equation}
 \log\!\left(
 \widetilde{\operatorname{gspar}}_\eps(f_F)+1
 \right)
 =\Theta_\eps\!\left(D(F)+\log\log(n+2)\right).
\label{eq:log-sparsity}
\end{equation}
If $D(F)=1$, then
$\operatorname{gspar}(f_F),\operatorname{gwt}(f_F)\le3$.
In particular, for every fixed integer $D_0\ge2$,
\[
 \widetilde{\operatorname{gspar}}_\eps(f_F)
 =\Theta_{\eps,D_0}(\log(n+1))
\]
whenever $D(F)=D_0$.
\end{corollary}

The next statement is the general structural result underlying the weight
lower bound.
For functions on the Boolean cube, use the unnormalized pairing
$\ip{g}{h}=\sum_xg(x)h(x)$.

\begin{theorem}[Dual exponential-profile transfer]
\label{thm:dual-transfer}
Let $\Rtree$ be a two-sided restriction tree for a Boolean function $f$, let
$0\le\eps<\eps'<1/2$, and define
\[
 k(\rho)=\widetilde{\deg}_{\eps'}(f|_\rho),
 \qquad
 \Phi_{\eps'}(f,\Rtree)
 =\mathbb E_\rho\!\left[
 2^{k(\rho)}\mathbf1[k(\rho)\ge1]
 \right].
\]
There is a signed function $\Psi$ on $\{0,1\}^n$ such that
\begin{equation}
 \max_{P,N}|\ip{\Psi}{m_{P,N}}|\le1,\qquad
 \|\Psi\|_1\le\Phi_{\eps'}(f,\Rtree),\qquad
 \ip{\Psi}{f}>\eps'\Phi_{\eps'}(f,\Rtree).
\label{eq:dual-witness}
\end{equation}
In particular,
\begin{equation}
 \widetilde{\operatorname{gwt}}_\eps(f)
 \ge(\eps'-\eps)\Phi_{\eps'}(f,\Rtree).
\label{eq:dual-weight-transfer}
\end{equation}
If $\Phi_{\eps'}(f,\Rtree)>0$, there is also a signed measure $\nu$ with
\begin{equation}
 \|\nu\|_1=1,\qquad
 \ip{\nu}{f}>\eps',\qquad
 \max_{P,N}|\ip{\nu}{m_{P,N}}|
 \le\frac1{\eps'\Phi_{\eps'}(f,\Rtree)}.
\label{eq:normalized-dual}
\end{equation}
\end{theorem}

\begin{corollary}[Paturi and quantum-query forms]
\label{cor:query-form}
For every fixed $0<\eps<1/2$ and every nonconstant symmetric Boolean
function,
\begin{equation}
 \log\!\left(
 \widetilde{\operatorname{gwt}}_\eps(f_F)+1
 \right)
 =\Theta_\eps\!\left(
 \frac{\widetilde{\deg}_{1/3}(f_F)^2}{n}
 \right)
 =\Theta_\eps\!\left(
 \frac{Q_{1/3}(f_F)^2}{n}
 \right).
\label{eq:query-weight}
\end{equation}
If $D(F)\ge2$, then
\begin{equation}
 \log\!\left(
 \widetilde{\operatorname{gspar}}_\eps(f_F)+1
 \right)
 =\Theta_\eps\!\left(
 \frac{Q_{1/3}(f_F)^2}{n}+\log\log(n+2)
 \right).
\label{eq:query-sparsity}
\end{equation}
\end{corollary}

\section{Exponential restriction-profile transfer}\label{sec:transfer}

We use the two-sided restriction trees of Chattopadhyay, Dahiya, and
Lovett~\cite{ChattopadhyayDahiyaLovett2026}.
A two-sided restriction tree may query a variable at most once on a
root-to-leaf path, and every internal node has outgoing labels $0,1,*$.
At a leaf, every variable not queried on the path is assigned a value in
$\{0,1\}$; this completion may depend on the leaf.
A random leaf chooses each outgoing branch with probability $1/3$.

For a generalized monomial $M$, define
\[
 Z_\rho(M)=
 \begin{cases}
 2^{\deg(M|_\rho)},&M|_\rho\not\equiv0,\\
 0,&M|_\rho\equiv0.
 \end{cases}
\]

\begin{lemma}[Exponential one-monomial estimate]
\label{lem:one-monomial}
For every two-sided restriction tree and every generalized monomial,
\[
 \mathbb E_\rho Z_\rho(M)\le1.
\]
\end{lemma}

\begin{proof}
Induct on the tree depth.
If the queried variable does not occur in $M$, average the three inductive
bounds.
If $M=x_iM'$, the branches $0,1,*$ contribute at most $0,1,2$ copies of the
potential of $M'$, so their average is one copy.
The case $M=(1-x_i)M'$ is identical.
At depth zero all variables are assigned and the potential is at most one.
\end{proof}

We first record the support part of the transfer.

\begin{lemma}[Primal profile transfer]
\label{lem:primal-transfer}
In the setting of \cref{thm:dual-transfer},
\[
 \widetilde{\operatorname{gspar}}_\eps(f)
 \ge\Phi_{\eps'}(f,\Rtree).
\]
\end{lemma}

\begin{proof}
Let $p=\sum_{j=1}^s\lambda_jM_j$ be an $\eps$-approximant.
At a leaf with $k=k(\rho)\ge1$, at least one surviving monomial has degree at
least $k$; otherwise $p|_\rho$ would be an $\eps'$-approximant of degree
below $k$.
Therefore
\[
 2^k\le\sum_{j=1}^s Z_\rho(M_j).
\]
Taking expectations and applying \cref{lem:one-monomial} gives
$\Phi_{\eps'}(f,\Rtree)\le s$.
\end{proof}

\begin{proof}[Proof of \cref{thm:dual-transfer}]
For every leaf with $k(\rho)\ge1$, linear-programming duality for
approximate degree gives a function $\phi_\rho$ on the residual cube such
that
\begin{equation}
 \|\phi_\rho\|_1=1,\qquad
 \ip{\phi_\rho}{f|_\rho}>\eps',\qquad
 \ip{\phi_\rho}{q}=0
 \quad\text{for every }\deg q<k(\rho).
\label{eq:leaf-dual}
\end{equation}
Let $L_\rho\phi_\rho$ denote its lift to the original cube, supported on the
subcube selected by the completed leaf restriction.
Define
\begin{equation}
 \Psi=
 \mathbb E_\rho\!\left[
 2^{k(\rho)}\mathbf1[k(\rho)\ge1]L_\rho\phi_\rho
 \right].
\label{eq:global-dual}
\end{equation}

For a generalized monomial $M$, the inner product
$\ip{\phi_\rho}{M|_\rho}$ vanishes whenever the restriction is zero or has
degree below $k(\rho)$.
Otherwise its absolute value is at most one.
Consequently,
\[
 |\ip{\Psi}{M}|
 \le
 \mathbb E_\rho\!\left[
 2^{k(\rho)}
 \mathbf1[\deg(M|_\rho)\ge k(\rho)]
 \right]
 \le\mathbb E_\rho Z_\rho(M)
 \le1.
\]
The triangle inequality and \eqref{eq:leaf-dual} give
\[
 \|\Psi\|_1\le\Phi_{\eps'}(f,\Rtree),
 \qquad
 \ip{\Psi}{f}>\eps'\Phi_{\eps'}(f,\Rtree),
\]
proving \eqref{eq:dual-witness}.

For any $\eps$-approximant $p=\sum_j\lambda_jM_j$,
\[
 \sum_j|\lambda_j|
 \ge\ip{\Psi}{p}
 \ge\ip{\Psi}{f}-\eps\|\Psi\|_1
 >(\eps'-\eps)\Phi_{\eps'}(f,\Rtree),
\]
which proves \eqref{eq:dual-weight-transfer}.
Finally, set $\nu=\Psi/\|\Psi\|_1$.
Since
\[
 \eps'\Phi_{\eps'}(f,\Rtree)
 <\ip{\Psi}{f}
 \le\|\Psi\|_1
 \le\Phi_{\eps'}(f,\Rtree),
\]
normalization gives \eqref{eq:normalized-dual}.
\end{proof}

The lifting formula \eqref{eq:global-dual} is explicit whenever the leaf
witnesses are explicit.
At a centered symmetric leaf one may replace $k(\rho)$ by any certified
lower degree $d(\rho)\le k(\rho)$ and use the explicit symmetric dual
witnesses of Bun and Thaler~\cite{BunThaler2015}.

\section{Centering a deepest transition}\label{sec:center}

We now build the restriction profile used for symmetric functions.

\begin{lemma}[Central completion]
\label{lem:central-completion}
Let $t\in\Tset(F)$ attain $D=D(F)$.
There is a two-sided restriction tree querying $D$ variables such that, if
$R$ variables remain free at a random leaf, then
$R\sim\operatorname{Bin}(D,1/3)$ and, for every leaf with $R=r\ge1$,
the residual symmetric function has a transition between weights
$\lceil r/2\rceil-1$ and $\lceil r/2\rceil$.
\end{lemma}

\begin{proof}
Fix a set $T$ of $D$ coordinates and query precisely those coordinates.
At a leaf, let $a,b,r$ be the numbers of queried coordinates fixed
respectively to one, fixed to zero, and left free.
Then $a+b+r=D$, and $R=r\sim\operatorname{Bin}(D,1/3)$.

For $r\ge1$, put $q=\lceil r/2\rceil$ and $c=t-a-q$.
Complete the $n-D$ coordinates outside $T$ by assigning exactly $c$ of them
to one.
This is feasible because
\[
 c\ge t-(D-r)-q=(t-D)+(r-q)\ge0,
\]
and
\[
 n-D-c=n-t-r-b+q
 \ge D-1-r-b+q=a+q-1\ge0.
\]
The residual profile is
\[
 F(a+c+|z|)=F(t-q+|z|),
\]
so the selected transition lies between residual weights $q-1$ and $q$.
\end{proof}

\begin{proof}[Proof of the lower bound in
\cref{thm:weight-classification}]
Fix $\eps<\eps'<1/2$.
At a leaf of \cref{lem:central-completion} with $R=r\ge1$, Paturi's
theorem gives
\[
 \widetilde{\deg}_{\eps'}(f_F|_\rho)
 \ge\alpha_{\eps'}r
\]
for a constant $\alpha_{\eps'}>0$.
Hence
\begin{align}
 \Phi_{\eps'}(f_F,\Rtree)
 &\ge
 \mathbb E\!\left[
 2^{\alpha_{\eps'}R}\mathbf1[R\ge1]
 \right]\notag\\
 &=
 \left(\frac{2+2^{\alpha_{\eps'}}}{3}\right)^D
 -\left(\frac23\right)^D.
\label{eq:symmetric-profile}
\end{align}
\Cref{thm:dual-transfer} now gives
$\widetilde{\operatorname{gwt}}_\eps(f_F)
=2^{\Omega_\eps(D)}$.
\Cref{lem:primal-transfer} gives the same exponential lower bound for
generalized sparsity.
\end{proof}

\subsection{Subcube discrepancy and conjunction edge}
\label{subsec:discrepancy}

Normalize the witness from \cref{thm:dual-transfer}, put
$\mu(x)=|\nu(x)|$, and let $y(x)=\operatorname{sgn}\nu(x)$ on the support of
$\mu$.
Then $\mu$ is a probability distribution and
\begin{equation}
 \left|
 \mathbb E_{x\sim\mu}[y(x)m_{P,N}(x)]
 \right|
 \le\frac1{\eps'\Phi_{\eps'}(f,\Rtree)}
\label{eq:atom-edge}
\end{equation}
for every conjunction of literals.
Since the constant function is itself a generalized monomial,
\begin{equation}
 \mathbb E_{x\sim\mu}[y(x)(2f(x)-1)]
 \ge
 2\eps'-\frac1{\eps'\Phi_{\eps'}(f,\Rtree)}.
\label{eq:target-edge}
\end{equation}
Writing a $\{-1,1\}$-valued conjunction hypothesis as
$h_{P,N}=2m_{P,N}-1$, we also have
\begin{equation}
 \left|
 \mathbb E_{x\sim\mu}[y(x)h_{P,N}(x)]
 \right|
 \le\frac3{\eps'\Phi_{\eps'}(f,\Rtree)}.
\label{eq:conjunction-edge}
\end{equation}

For the deepest-transition tree,
\eqref{eq:symmetric-profile} makes the right side of
\eqref{eq:conjunction-edge} equal to
$2^{-\Omega_{\eps'}(D(F))}$, whereas
\eqref{eq:target-edge} is a positive constant.
Thus there is an agnostic labeled distribution on which the symmetric
predicate has constant edge but every individual conjunction has
exponentially small edge.
This is a weak-learning obstruction, not a computational lower bound for
arbitrary learning algorithms.

\section{Symmetric atoms and exponential profiles}\label{sec:atoms}

For $a,b\ge0$ with $a+b\le n$, define
\[
 B_{a,b}(w)=
 \frac{\binom wa\binom{n-w}b}
 {\binom na\binom{n-a}b}.
\]
This is the uniform average of all generalized monomials with $a$ positive
and $b$ negative literals.

\begin{lemma}[Symmetric atomic reduction]
\label{lem:symmetric-reduction}
For every symmetric real function $g(x)=G(|x|)$,
\begin{multline}
 \widetilde{\operatorname{gwt}}_\eps(g)
 =\min\Biggl\{
 \sum_{a+b\le n}|c_{a,b}|:\ 
 \left|
 G(w)-\sum_{a+b\le n}c_{a,b}B_{a,b}(w)
 \right|\le\eps\\
 \text{for every }w\in\{0,\ldots,n\}
 \Biggr\}.
\label{eq:atomic-reduction}
\end{multline}
\end{lemma}

\begin{proof}
Average an arbitrary approximating representation over all coordinate
permutations.
The orbit average of a monomial with $a$ positive and $b$ negative literals
is $B_{a,b}$, and grouping coefficients cannot increase their total absolute
mass.
Conversely, each $B_{a,b}$ is an average of generalized monomials with total
coefficient mass one, so every feasible univariate combination lifts with
the same mass.
\end{proof}

\begin{lemma}[Unit-weight exponential profiles]
\label{lem:exponential-atom}
For every $r\in(0,1]$, the profiles $r^w$ and $r^{n-w}$ have generalized
weight at most one.
\end{lemma}

\begin{proof}
For the first profile,
\begin{equation}
 \sum_{b=0}^n
 \binom nb(1-r)^br^{n-b}B_{0,b}(w)=r^w.
\label{eq:exponential-identity}
\end{equation}
The coefficients are nonnegative and sum to one.
Complementing every input bit proves the assertion for $r^{n-w}$.
\end{proof}

\begin{lemma}[Bounded polynomial coefficient mass]
\label{lem:coefficient-mass}
If $p(z)=\sum_{j=0}^m\alpha_jz^j$ has degree at most $m$ and
$\|p\|_{L_\infty([0,1])}\le B$, then
\begin{equation}
 \sum_{j=0}^m|\alpha_j|\le4B\,7^m.
\label{eq:coefficient-mass}
\end{equation}
\end{lemma}

\begin{proof}
Expand $p$ in the shifted Chebyshev basis
$T_j(2z-1)$.
The zeroth Chebyshev coefficient has magnitude at most $B$, and every other
coefficient has magnitude at most $2B$.
The monomial coefficient $\ell_1$-norm of $T_j(2z-1)$ is at most $7^j$:
this follows inductively from the Chebyshev recurrence and the fact that
$2(2z-1)$ has coefficient mass six.
Summing the resulting geometric series gives \eqref{eq:coefficient-mass}.
\end{proof}

\section{Edge approximation and error amplification}
\label{sec:amplification}

The next lemma approximates the entire edge pattern at once.

\begin{lemma}[Constant-error exponential edge profile]
\label{lem:constant-edge}
There are absolute constants $C>0$ and $0<\eta_0<1/100$ with the following
property.
Let $D\ge1$, and let $u:\{0,1,\ldots,D\}\to[-1,1]$ satisfy $u(D)=0$.
There are $m\le CD$, coefficients $\alpha_0,\ldots,\alpha_m$, and rates
$r_j=e^{-j/D}$ such that
\[
 q(w)=\sum_{j=0}^m\alpha_jr_j^w
\]
satisfies
\begin{equation}
 |q(w)-u(w)|\le\eta_0\quad(0\le w<D),
 \qquad
 |q(w)|\le\eta_0\quad(w\ge D),
\label{eq:constant-edge-error}
\end{equation}
and
\begin{equation}
 \sum_{j=0}^m|\alpha_j|\le2^{CD}.
\label{eq:constant-edge-mass}
\end{equation}
\end{lemma}

\begin{proof}
Put $\sigma=e^{-1/D}$ and $z_w=\sigma^w$.
Define $h:[0,1]\to[-1,1]$ by setting $h=0$ on $[0,z_D]$, requiring
$h(z_w)=u(w)$ for $0\le w\le D$, and interpolating linearly between
successive nodes.
For $0\le w<D$,
\[
 z_w-z_{w+1}
 =e^{-w/D}(1-e^{-1/D})
 \ge\frac cD
\]
for an absolute $c>0$.
Thus $h$ is $CD$-Lipschitz.
The algebraic Jackson inequality~\cite{DeVoreLorentz1993} gives a polynomial
$p$ of degree at most $CD$ such that
$\|p-h\|_{L_\infty([0,1])}\le\eta_0$.
Write $p(z)=\sum_{j=0}^m\alpha_jz^j$.
Since $\|p\|_\infty\le2$, \cref{lem:coefficient-mass} gives
\eqref{eq:constant-edge-mass}.
Finally,
\[
 p(\sigma^w)
 =\sum_{j=0}^m\alpha_j(\sigma^j)^w,
\]
and $h(\sigma^w)=u(w)$ for $w<D$, while $h(\sigma^w)=0$ for every integer
$w\ge D$.
\end{proof}

The edge corrections used below take values in $\{-1,0,1\}$.
We next reduce the constant error in \cref{lem:constant-edge}
exponentially in the degree of an output polynomial.

\begin{lemma}[Ternary output amplifier]
\label{lem:ternary-amplifier}
There is an absolute constant $C>0$ with the following property.
Let $u:X\to\{-1,0,1\}$ and $q:X\to\mathbb R$ satisfy
$\|q-u\|_\infty\le\eta_0$.
For every $0<\delta<1/2$, there is a univariate polynomial $A$ of degree at
most $C\log(2/\delta)$ such that
\[
 \|A\circ q-u\|_\infty\le\delta.
\]
If
\[
 q(w)=\sum_{j=1}^L\alpha_jr_j^w,
 \qquad
 W=\max\left\{1,\sum_j|\alpha_j|\right\},
\]
then $A(q(w))$ is a signed exponential sum whose rates are products of the
$r_j$, with coefficient mass and number of terms at most
\begin{equation}
 (C(W+1))^{C\log(2/\delta)}
 \quad\text{and}\quad
 (C(L+1))^{C\log(2/\delta)},
\label{eq:amplifier-complexity}
\end{equation}
respectively.
\end{lemma}

\begin{proof}
For odd $k$, let
\[
 H_k(t)=
 \sum_{j=(k+1)/2}^k\binom kjt^j(1-t)^{k-j}.
\]
Chernoff's inequality gives
\[
 H_k(t)\le e^{-ck}\quad(0\le t\le1/5),
 \qquad
 H_k(t)\ge1-e^{-ck}\quad(4/5\le t\le1)
\]
for an absolute constant $c>0$.

Put
\[
 s_+(z)=\frac{z(z+1)}2,\qquad
 s_-(z)=\frac{z(z-1)}2,
\]
and
\[
 t_\pm(z)=\frac1{10}+\frac45s_\pm(z).
\]
For a sufficiently small absolute $\eta_0$, if $z$ lies within $\eta_0$ of
$\{-1,0,1\}$, then each $t_\pm(z)$ lies in $[0,1]$.
Moreover, $t_+(z)\ge4/5$ precisely near $1$ and is at most $1/5$ near
$0,-1$; the analogous statement holds for $t_-$ and the neighborhood of
$-1$.
Thus
\[
 A(z)=H_k(t_+(z))-H_k(t_-(z))
\]
has error at most $2e^{-ck}$ on the three neighborhoods.
Taking $k=O(\log(2/\delta))$ proves the approximation claim.

The polynomial $A$ has degree at most $2k$ and coefficient mass
$2^{O(k)}$ in its displayed arithmetic construction.
The exponential coefficient mass of each $t_\pm(q)$ and $1-t_\pm(q)$ is at
most $C(W+1)^2$, and the number of their exponential terms is at most
$C(L+1)^2$.
Expanding the Bernstein sum for $H_k$ proves
\eqref{eq:amplifier-complexity}.
Products stay inside the exponential dictionary because
$r_i^wr_j^w=(r_ir_j)^w$.
\end{proof}

\begin{corollary}[Small-error exponential edge profile]
\label{cor:small-edge}
Let $u:\{0,\ldots,D\}\to\{-1,0,1\}$ satisfy $u(D)=0$.
For every $0<\delta<1/2$, there is a signed exponential sum $q_\delta$
satisfying
\begin{equation}
 |q_\delta(w)-u(w)|\le\delta\quad(0\le w<D),
 \qquad
 |q_\delta(w)|\le\delta\quad(w\ge D),
\label{eq:small-edge-error}
\end{equation}
whose coefficient mass and number of terms are at most
\begin{equation}
 2^{CD\log(2/\delta)}.
\label{eq:small-edge-complexity}
\end{equation}
Every nonconstant rate in the sum has the form $e^{-s/D}$ for an integer
$s\ge1$.
\end{corollary}

\begin{proof}
\Cref{lem:constant-edge} gives $W\le2^{O(D)}$ and $L=O(D)$.
Apply \cref{lem:ternary-amplifier} and use
\[
 (CD)^{O(\log(2/\delta))}
 \le 2^{O(D\log(2/\delta))}.
\]
Products of the base rates $e^{-j/D}$ have the asserted form.
\end{proof}

\section{Sparse realization of exponential atoms}\label{sec:sparsification}

Identity \eqref{eq:exponential-identity} realizes $r^w$ with unit weight but
may use many subcubes.
The next lemma gives a uniformly accurate sparse realization.

\begin{lemma}[Empirical exponential atom]
\label{lem:empirical-atom}
There is an absolute constant $C>0$ such that, for every $n\ge1$,
$r\in(0,1)$, and $0<\delta<1/2$, there are subsets
$S_1,\ldots,S_L\subseteq[n]$ with
\begin{equation}
 L\le
 C\delta^{-2}
 \left[
 \left(1+\frac{\ln(1/\delta)}{-\ln r}\right)\log(n+1)+1
 \right]
\label{eq:empirical-size}
\end{equation}
such that
\[
 R(x)=\frac1L\sum_{\ell=1}^L
 \prod_{i\in S_\ell}(1-x_i)
\]
satisfies
\begin{equation}
 \|R-r^{|x|}\|_\infty\le2\delta.
\label{eq:empirical-error}
\end{equation}
In particular, $\operatorname{gspar}(R)\le L$ and
$\operatorname{gwt}(R)\le1$.
The complemented construction approximates $r^{n-|x|}$ with the same
bounds.
\end{lemma}

\begin{proof}
Sample each $S_\ell$ independently by including every coordinate with
probability $1-r$.
If $X=\{i:x_i=1\}$, then
\[
 \mathbb E\prod_{i\in S_\ell}(1-x_i)
 =\Pr[S_\ell\cap X=\varnothing]
 =r^{|X|}.
\]
Put
\[
 K=\left\lceil\frac{\ln(1/\delta)}{-\ln r}\right\rceil.
\]
There are at most $(n+1)^K$ inputs of weight at most $K$.
Hoeffding's inequality and a union bound show that $L$ as in
\eqref{eq:empirical-size} can be chosen so that
\[
 |R(x)-r^{|x|}|\le\delta
\]
simultaneously on all such inputs.

If $|X|>K$, fix a $K$-subset $T\subseteq X$.
Every sampled monomial that survives on $x$ also survives on $1_T$, so
\[
 0\le R(x)\le R(1_T)\le r^K+\delta\le2\delta.
\]
Also $0\le r^{|X|}\le r^K\le\delta$, proving
\eqref{eq:empirical-error}.
The weight assertion follows because $R$ is an average of generalized
monomials.
Complementing all input bits proves the final statement.
\end{proof}

\begin{lemma}[Sparsifying an amplified exponential sum]
\label{lem:sparse-sum}
Suppose
\[
 q(w)=\sum_{j=1}^A\alpha_jr_j^w
\]
approximates a lower-edge profile to error $\delta$, has coefficient mass
$W\ge1$, and every nonconstant rate satisfies $-\ln r_j\ge1/D$.
If
\[
 A,W\le2^{CD\log(2/\delta)},
\]
then $q(|x|)$ can be replaced by a generalized polynomial $Q$ with total
error at most $2\delta$ and
\begin{equation}
 \operatorname{gspar}(Q)
 \le2^{C'D\log(2/\delta)}\log(n+1),
 \qquad
 \operatorname{gwt}(Q)\le W.
\label{eq:sparse-sum}
\end{equation}
The same holds for an upper-edge profile in the coordinate $n-w$.
\end{lemma}

\begin{proof}
Keep the rate-one term exactly.
For every other term, apply \cref{lem:empirical-atom} with
\[
 \zeta=\frac{\delta}{4W}.
\]
The replacement error is at most $2\zeta W\le\delta/2$, and the weight does
not increase.
Moreover,
\[
 \ln(1/\zeta)=O(D\log(2/\delta)),
 \qquad
 \frac{\ln(1/\zeta)}{-\ln r_j}
 =O(D^2\log(2/\delta)).
\]
Summing \eqref{eq:empirical-size} over the $A$ profiles and absorbing
polynomial factors in $D$ and $\log(2/\delta)$ into the exponential proves
\eqref{eq:sparse-sum}.
The upper-edge statement follows by complementing all input bits.
\end{proof}

\section{Proof of the upper bounds}\label{sec:upper-proof}

\begin{proof}[Proof of the upper bound in
\cref{thm:weight-classification} and of \cref{thm:sparsity-upper}]
Let $D=D(F)$.
The case $D=0$ is immediate.
If $n<2D$, use the exact singleton representation, whose support and weight
are at most $2^n\le2^{2D}$.
We may therefore assume that $n\ge2D\ge2$.
There is no transition at any
$t\in\{D+1,\ldots,n-D\}$, so $F$ is constant on
$\{D,D+1,\ldots,n-D\}$.
Let this value be $\beta$.
Define two edge profiles on $\{0,1,\ldots,D\}$:
\[
\begin{aligned}
 u_-(w)&=
 \begin{cases}
 F(w)-\beta,&0\le w<D,\\
 0,&w=D,
 \end{cases}\\
 u_+(w)&=
 \begin{cases}
 F(n-w)-\beta,&0\le w<D,\\
 0,&w=D.
 \end{cases}
\end{aligned}
\]
Both take values in $\{-1,0,1\}$.

Apply \cref{cor:small-edge} to $u_-$ in the coordinate $w$, and to $u_+$ in
the coordinate $n-w$, each with error $\eps/2$.
Their sum with the constant $\beta$ approximates $F$ to error at most
$\eps$.
\Cref{lem:exponential-atom} and \eqref{eq:small-edge-complexity} give
\[
 \widetilde{\operatorname{gwt}}_\eps(f_F)
 \le2^{O(D\log(2/\eps))}.
\]

For sparsity, apply \cref{lem:sparse-sum} to the two amplified edge profiles
with error parameter $\eps/4$.
Each edge then has total error at most $\eps/2$, so their sum with $\beta$
has error at most $\eps$.
The resulting approximant satisfies
\[
 \operatorname{gspar}
 \le2^{O(D\log(2/\eps))}\log(n+1),
 \qquad
 \operatorname{gwt}
 \le2^{O(D\log(2/\eps))}.
\]

Finally suppose $D=1$.
For $n\ge2$, $F$ is constant, say equal to $\beta$, on weights
$1,\ldots,n-1$, and
\[
 f_F(x)=
 \beta
 +(F(0)-\beta)\prod_{i=1}^n(1-x_i)
 +(F(n)-\beta)\prod_{i=1}^nx_i.
\]
Thus exact generalized sparsity and weight are at most three.
The case $n=1$ is immediate.
\end{proof}

\section{Coordinate signatures}\label{sec:signatures}

\begin{proof}[Proof of \cref{thm:signature-lower}]
Let
\[
 p(x)=\sum_{j=1}^sa_jm_{P_j,N_j}(x)
\]
be any $\eps$-approximant to $f_F$.
Associate with each coordinate $i\in[n]$ its signature
\[
 \tau(i)\in\{+,-,0\}^s,
\]
whose $j$th entry records whether $i\in P_j$, $i\in N_j$, or neither.
There are at most $3^s$ signatures.

Choose a transition $t$ attaining depth $D$.
Then either $t=D$ or $t=n-D+1$.
Suppose some signature class $C$ has more than $D$ coordinates.
First consider $t=D$.
Choose $A\subset B\subset C$ with $|A|=D-1$ and $|B|=D$, and compare the
inputs $1_A$ and $1_B$.

Consider one generalized monomial.
If the common signature of $C$ in this monomial is $0$, it has the same
value on the two inputs.
If the signature is $+$, some coordinate of $C\setminus B$ is a positive
literal set to zero, so the monomial vanishes on both inputs.
If the signature is $-$, both $A$ and $B$ are nonempty because $D\ge2$, so
a negative literal is falsified on both inputs.
Thus every monomial, and hence $p$, has the same value on $1_A$ and $1_B$.
This contradicts $F(D-1)\ne F(D)$ and $\eps<1/2$.

If $t=n-D+1$, use the complemented inputs whose zero sets are $A$ and $B$.
The same argument applies with positive and negative literals interchanged.
Therefore every signature class has size at most $D$, so
\[
 n\le D\,3^s,
\]
which proves \eqref{eq:signature-lower}.
\end{proof}

\begin{proof}[Proof of \cref{cor:log-sparsity}]
The restriction-profile lower bound and \cref{thm:signature-lower} give
\[
 \log\!\left(
 \widetilde{\operatorname{gspar}}_\eps(f_F)+1
 \right)
 \ge\Omega_\eps(D)
\]
and
\[
 \log\!\left(
 \widetilde{\operatorname{gspar}}_\eps(f_F)+1
 \right)
 \ge\Omega(\log\log(n/D+2)).
\]
If $D\ge\frac12\log\log(n+2)$, the first bound dominates.
Otherwise $\log(n/D)=\Theta(\log n)$, and the second bound dominates.
The upper side follows by taking logarithms in
\eqref{eq:sparsity-upper}.
The $D=1$ representation was proved in \cref{sec:upper-proof}.
\end{proof}

\section{Quantum query complexity}\label{sec:quantum}

Paturi's theorem gives
\begin{equation}
 \widetilde{\deg}_{1/3}(f_F)
 =\Theta(\sqrt{nD(F)}).
\label{eq:paturi-form}
\end{equation}
Bounded-error quantum query complexity has the same asymptotic form:
\begin{equation}
 Q_{1/3}(f_F)=\Theta(\sqrt{nD(F)}).
\label{eq:quantum-form}
\end{equation}
The lower bound follows from the polynomial method
of Beals et al.~\cite{BealsEtAl2001}: the acceptance probability of a
$T$-query algorithm is a polynomial of degree at most $2T$.
For completeness, the upper bound follows from quantum approximate
counting~\cite{BrassardHoyerMoscaTapp2002}.
With $T=C\sqrt{nD}$ queries, its estimate $\widetilde w$ of the number $w$
of one-coordinates obeys, with constant success probability,
\begin{equation}
 |\widetilde w-w|
 =O\!\left(
 \frac{\sqrt{w(n-w)}}T+\frac n{T^2}
 \right).
\label{eq:counting-error}
\end{equation}
For a sufficiently large absolute constant $C$, this recovers $w$ exactly
when $w<D$ and distinguishes that case from $w\ge D$.
Apply the same procedure to the zero-coordinates.
It recovers the exact Hamming weight whenever it lies within distance
$D-1$ of an endpoint; otherwise it certifies that the weight lies in the
transition-free middle interval, where $F$ is constant.
A constant number of repetitions makes the two calls simultaneously correct
with probability at least $2/3$, at total cost $O(\sqrt{nD})$.

\begin{proof}[Proof of \cref{cor:query-form}]
Equation \eqref{eq:paturi-form} and
\cref{thm:weight-classification} give the first equality in
\eqref{eq:query-weight}; \eqref{eq:quantum-form} gives the second.
Combining \eqref{eq:quantum-form} with \cref{cor:log-sparsity} proves
\eqref{eq:query-sparsity}.
\end{proof}

\section{Scope and open problems}\label{sec:scope}

\Cref{thm:weight-classification} completely classifies approximate
signed-subcube weight for symmetric Boolean functions.
\Cref{cor:log-sparsity} gives a complete classification of approximate
signed-subcube sparsity after taking logarithms.
The error-amplification argument improves the uniform upper dependence from
$2^{O(D/\eps)}$ to $(2/\eps)^{O(D)}$, and
\cref{thm:dual-transfer} supplies a general dual witness with a
subcube-discrepancy interpretation.

Several gaps remain.
First, on the original support scale the bounds are
\[
 \max\{2^{\Omega_\eps(D)},\log_3(n/D)\}
 \le
 \widetilde{\operatorname{gspar}}_\eps(f_F)
 \le
 2^{O(D\log(2/\eps))}\log(n+1).
\]
Thus there is a genuine product-versus-maximum gap even though the logarithms
match.
Determining the original-scale answer for a natural family such as
$\operatorname{EXACT}_D$ or a single interval remains open.

Second, the upper error dependence is polynomial in $1/\eps$ with exponent
$O(D)$, but no matching lower dependence is known.
Exact representation reintroduces dimension dependence, so the correct
interpolation as $\eps\downarrow0$ may involve all three parameters
$D,n,\eps$.

Third, the coefficient estimate $4B\,7^m$ is convenient rather than sharp.
Optimizing the exponential base and giving simpler closed-form global dual
witnesses are worthwhile quantitative problems.

Finally, the transfer theorem is not limited to symmetric functions.
Its present application uses symmetric sections because Paturi's theorem
provides linear approximate degree at every centered leaf.
Finding a genuinely nonsymmetric class with an exponential restriction
profile not explained by a hard symmetric section is the main structural
problem.
On the applications side, \cref{subsec:discrepancy} gives a precise
obstruction for conjunction weak learners.
Turning it into a computational lower bound in a standard learning model, or
into a communication-complexity lifting theorem, would provide a stronger
external consequence.

\section*{Acknowledgments}

During the preparation of this work, the author used OpenAI's GPT model as an interactive research assistant to explore proof strategies, check intermediate derivations, and help refine theorem statements and proofs, as well as to improve the readability, structural flow, and \LaTeX{} formatting of the manuscript. 

\bibliographystyle{alpha}
\bibliography{references}

\end{document}